\documentclass[conference]{IEEEtran}
\usepackage{comment}
\usepackage{dependencies}
\IEEEoverridecommandlockouts
\usepackage{cite}
\usepackage{braket}
\usepackage{amsmath,amssymb,amsfonts,amsthm}
\newtheorem{theorem}{Theorem}[]
\newtheorem{lemma}[theorem]{Lemma}
\usepackage{algorithmic}
\usepackage{graphicx}
\usepackage{textcomp}
\usepackage{xcolor}

\usepackage{hyperref}

\def\BibTeX{{\rm B\kern-.05em{\sc i\kern-.025em b}\kern-.08em
    T\kern-.1667em\lower.7ex\hbox{E}\kern-.125emX}}
    
\begin{document}

\title{Optimal Scaling Quantum Interior Point Method for Linear Optimization 
}

\author{
\centering
\IEEEauthorblockN{Mohammadhossein Mohammadisiahroudi}
\IEEEauthorblockA{\textit{Mathematics and Statistics} \\
\textit{Quantum Science Institute} \\
\textit{University of Maryland, Baltimore County}\\
Baltimore, MD \\
mhms379@umbc.edu}

\and

\IEEEauthorblockN{Zeguan Wu}
\IEEEauthorblockA{\textit{Computer Science} \\
\textit{University of Pittsburgh}\\
Pittsburgh, PA \\
zew79@pitt.edu}

\and

\IEEEauthorblockN{Pouya Sampourmahani}
\IEEEauthorblockA{\textit{Industrial and Systems Engineering} \\
\textit{Lehigh University}\\
Bethlehem, PA \\
pos220@lehigh.edu}

\and

\IEEEauthorblockN{Jun-Kai You}
\IEEEauthorblockA{\textit{Industrial and Systems Engineering} \\
\textit{Lehigh University}\\
Bethlehem, PA \\
juy324@lehigh.edu}

\and

\IEEEauthorblockN{Tam\'as Terlaky}
\IEEEauthorblockA{\textit{Industrial and Systems Engineering} \\
\textit{Lehigh University}\\
Bethlehem, PA \\
terlaky@lehigh.edu}
}

\IEEEpubid{\begin{minipage}[t]{\textwidth}\ \\[10pt]
\centering
979-8-3315-5736-2/25/\$31.00 \copyright 2025 IEEE. Personal use of this material is permitted. Permission from IEEE must be obtained for all other uses, in any current or future media, including reprinting/republishing this material for advertising or promotional purposes, creating new collective works, for resale or redistribution to servers or lists, or reuse of any copyrighted component of this work in other works. DOI 10.1109/QCE65121.2025.00044
\end{minipage}}

\maketitle

\begin{abstract}
The emergence of huge-scale, data-intensive linear optimization (LO) problems in applications such as machine learning has driven the need for more computationally efficient interior point methods (IPMs). While conventional IPMs are polynomial-time algorithms with rapid convergence, their per-iteration cost can be prohibitively high for dense large-scale LO problems. Quantum linear system solvers have shown potential in accelerating the solution of linear systems arising in IPMs. 
In this work, we introduce a novel almost-exact quantum IPM, where the Newton system is constructed and solved on a quantum computer, while solution updates occur on a classical machine. Additionally, all matrix-vector products are performed on the quantum hardware. This hybrid quantum-classical framework achieves an optimal worst-case scaling of $\mathcal{O}(n^2)$ for fully dense LO problems. To ensure high precision, despite the limited accuracy of quantum operations, we incorporate iterative refinement techniques both within and outside the proposed IPM iterations. 
The proposed algorithm has a quantum complexity of $\mathcal{O}(n^{1.5} \kappa_A \log(\frac{1}{\epsilon}))$ queries to QRAM and $\mathcal{O}(n^2 \log(\frac{1}{\epsilon}))$ classical arithmetic operations. Our method outperforms the worst-case complexity of prior classical and quantum IPMs, offering a significant improvement in scalability and computational efficiency.
\end{abstract}


\section{Introduction}
The standard form linear optimization (LO) problem is  minimizing a linear objective function over a polyhedron, formally defined as
\begin{equation} \label{eq:primall problem}\tag{P}
    \begin{aligned}
    \min_{x\in \mathbb{R}^{n}}\  c^T&x \\
    \text{s.t. }
    Ax &= b, \\
    x &\geq 0,
    \end{aligned}
\end{equation}
where $A\in \mathbb{R}^{m\times n}, b\in\mathbb{R}^m$, and $c\in\mathbb{R}^n$. It is well-known that there is a dual problem associated with the primal problem, as
\begin{equation} \label{eq:dual problem}\tag{D}
    \begin{aligned}
    \max_{(y,s)\in \mathbb{R}^{m}\times\mathbb{R}^{n}} \  b^Ty\ \ & \\
    \text{s.t. }
    A^Ty +&s = c,\\
    &s \geq 0.
    \end{aligned}
\end{equation}
By the strong duality theorem \cite{roos2005interior}, all optimal solutions, if they exist, belong to the set $\mathcal{PD}^*$, which is defined as
\begin{align*}
\mathcal{PD}^*=&\{(x,y,s)\in\mathbb{R}^{n+m+n}:\ Ax=b,\ A^Ty+s=c,\\ & x^Ts=0, \ (x,s)\geq0 \}.
\end{align*}

One of the first prevailing methods for finding an optimal solution for LO problems is the Simplex algorithm\cite{bertsimas1997introduction}. Despite their efficiency in many practical problems, simplex methods can take exponential time to solve the problem in the worst case \cite{klee1972good}.

The most efficient methods for solving large-scale LO problems are interior point methods (IPMs). The era of IPMs was launched by Karmarkar’s projective method for LO problems\cite{Karmarkar1984_New}. Many variants of IPMs have been proposed for LO problems and nonlinear optimization problems\cite{roos2005interior,nesterov1994interior}. 
When an LO problem has feasible interior solutions, it is proved that its perturbed optimality conditions define a unique analytic curve, called the central path, where the limiting point is an optimal solution. IPMs typically start from a point close to the central path and use Newton's method to iteratively track the central path until the solution is close enough to the limiting point. 
To get an $\epsilon$-approximate solution $x$ to the LO problem such that $x^Ts\leq \epsilon$, IPMs need $\mathcal{O}(\sqrt{n}\log(1/\epsilon))$ iterations\cite{roos2005interior}. However, the bottleneck of IPMs lies in the cost of each IPM iteration.

In each IPM iteration, a Newton linear system needs to be solved.
To achieve the $\mathcal{O}(\sqrt{n}\log(1/\epsilon))$ IPM iteration complexity, the Newton step obtained from the Newton linear system solution has to be feasible. Theoretically, one can use direct methods, including Cholesky factorization, to accurately solve the Newton linear system, which maintains feasibility easily. However, such direct methods have $\mathcal{O}(n^3)$ complexity and are thus not applicable to large-scale problems.
To handle large-scale problems, people utilize iterative methods for solving Newton linear systems, including conjugate gradient (CG) methods~\cite{al2009convergence, Monteiro2003_Convergence}. These iterative methods have $\mathcal{O}(ns\kappa \frac{1}{\epsilon})$ complexity for the Newton linear systems arising from IPMs for LO problems, where $\kappa$ is the condition number of the matrix of the Newton system and $s$ denotes the number of nonzero elements at each row of the coefficient matrix. Although iterative methods have a milder dependence on dimension than direct methods, they require an inexact framework of IPMs and preconditioning techniques to handle errors and moderate condition number dependence.

To improve the worst-case complexity of IPMs, partial update techniques were deployed to calculate the inexact Newton direction via a few rank-one updates for the inverse of the NES matrix. This approach leads to the best total complexity of $\Ocal(n^3L)$ arithmetic operations for solving LO problems \cite{roos2005interior}. 
Recently, this idea has been fortified by using techniques like fast matrix multiplication, spectral sparsification, and stochastic central path methods \cite{lee2015efficient, cohen2021solving}. By utilizing these approximations, the complexity of IPMs can be improved to $\Ocal(n^{\omega}\log(\frac{n}{\epsilon}))$, where $\epsilon$ is the target optimality gap and $\omega <2.3729$  is the matrix multiplication constant \cite{brand2020}.

In another direction, first-order methods are adopted for solving large-scale LO problems. One of the recent algorithms is the primal-dual hybrid gradient (PDHG) method. PDHG methods are primarily developed for convex optimization problems, focusing on imaging problems\cite{chambolle2011first}. Recently, Applegate et al. proposed a practical PDHG method for LO problems, showing its potential in solving certain large-scale LO problems as well as commercial solvers~\cite{applegate2021practical}. Note that there is no theoretical complexity bound for these algorithms.

Recently, quantum computing has demonstrated increasing potential in speeding up certain algorithms. Quantum linear system algorithms (QLSAs) can solve linear systems faster than classical algorithms under certain conditions \cite{harrow2009quantum,childs2017quantum}. Researchers have been developing quantum IPMs (QIPMs) by leveraging the power of QLSAs to solve Newton linear systems \cite{kerenidis2020quantum}.
QIPMs have been developed for LO \cite{mohammadisiahroudi2025improvements,apers2023quantum}, convex quadratic optimization\cite{wu2023inexact}, second-order conic optimization\cite{augustino2021inexact}, and semidefinite optimization\cite{kerenidis2020quantum,augustino2023quantum,mohammadisiahroudi2025quantum}. 
For LO problems, the complexity of QIPMs has been improved through a line of research \cite{mohammadisiahroudi2024efficient, mohammadisiahroudi2025improvements, augustino2021inexact} and most recently, Wu et al. introduced a QIPM based on dual Logarithmic barrier method~\cite{wu2024quantum}, with $\mathcal{O}(\sqrt{n}\log(n/\epsilon))$ IPM iteration complexity. In each IPM iteration, their QIPM needs almost $\mathcal{O}(n)$ query to QRAM and $\mathcal{O}(n^2)$ 
classical arithmetic operations. It is shown that the most expensive step in the QIPM is the matrix-vector multiplication, which is implemented on classical computers. In this work, we propose implementing matrix-vector multiplication steps on quantum computers to further improve the performance of QIPMs. Additionally, we introduce almost-exact IPMs, which are novel types of IPM achieving best-known iteration complexity $\Ocal(\sqrt{n}\log(\frac{1}{\epsilon}))$ using exponentially precise solutions for Newton Systems. To achieve exponentially small errors, we use iterative refinement techniques inside and outside of the proposed QIPM.

\subsection{Contributions}
\begin{itemize}
\item We propose a novel almost-exact  QIPM that achieves an optimal runtime scaling of $\mathcal{O}(n^2)$, improving upon both classical and existing quantum IPMs in terms of dimensional complexity.

\item Our framework supports inexact quantum operations, including quantum matrix inversion, matrix-vector multiplication, and matrix-matrix multiplication, by incorporating iterative refinement. This capability marks a significant step toward developing a fully quantum IPM.

\item In contrast to prior QIPMs, our method eliminates the need for any classical matrix operations, resulting in a total classical arithmetic cost of only $\mathcal{O}\left(n^2 \log\left(\frac{1}{\epsilon}\right)\right)$. This represents an asymptotic improvement of $\mathcal{O}(\sqrt{n})$ over previous QIPM approaches.

\item Notably, in most existing QIPMs, replacing QLSA with classical methods such as Conjugate Gradient yields comparable complexity, thus offering no definitive quantum speedup. In contrast, our approach demonstrates a provable quantum advantage, as any classical counterpart would require at least $\mathcal{O}(n^{2.5})$ operations.
\end{itemize}

The remainder of the paper is organized as follows. Section~\ref{sec: AE} introduces the proposed almost-exact framework for quantum interior point methods. Section~\ref{sec: Sub} describes the quantum subroutines used within the framework. In Section~\ref{sec: IR}, we show how iterative refinement is incorporated to achieve high precision. Section~\ref{sec: total} analyzes the overall computational complexity of our approach. Finally, Section~\ref{sec: con} concludes the paper.

\section{Almost-exact Quantum interior point method}\label{sec: AE}
In this section, we propose an almost exact quantum interior point method for solving linear optimization problems. Assuming that the input data is all integer, we denote the binary length of the input data by 
\begin{align*}
    L&=mn+m+n+\sum_{i,j}\lceil\log_2(|a_{ij}|+1)\rceil\\
&+\sum_{i}\lceil\log_2(|c_{i}|+1)\rceil+\sum_{j}\lceil\log_2(|b_{j}|+1)\rceil,
\end{align*}
where $a_{ij}$ represents the $ij$-element of matrix $A$. The optimal partition is also defined as 
\begin{align*}
    \Bcal&=\{j\in\{1,\dots,n\}:x^*_j>0 \text{ for some }(x^*,y^*, s^*)\in \mathcal{PD}^*\},\\
    \Ncal&=\{j\in\{1,\dots,n\}:s^*_j>0 \text{ for some }(x^*,y^*, s^*)\in \mathcal{PD}^*\}.
\end{align*}
The following lemma is a classical result first proved by \cite{Khachiyan1980}.
\begin{lemma}\label{lemma: L bound}
Let $(x^*,y^*,s^*)\in \mathcal{PD}^*$ be a basic solution. If $x_i^*>0$, then we have $x_i^*\geq 2^{-L}$. If $s_i^*>0$, then we have $s_i^*\geq 2^{-L}$.
\end{lemma}

Lemma~\ref{lemma: L bound} is a fundamental result in the complexity analysis of IPMs. It means that after a sufficient number of iterations of IPMs, a decision variable can be rounded to zero if it is smaller than $2^{-L}$. Then, by a rounding procedure, one can find an exact optimal solution for linear optimization \cite{roos2005interior, wright1997primal}. In the proposed algorithm, all calculations happen on a quantum machine with precision $\epsilon=2^{-tL}$ where $t$ is a small constant, less than $10$. This high level of accuracy justifies describing the algorithm as \emph{almost-exact}. The only calculation that happens on a classical computer is updating the solution and vector-vector summation. In this paper, we use the dual logarithmic barrier method, which has a simple framework. At each step of the dual log barrier IPM, we need to solve the following Newton system
\begin{equation}
\begin{bmatrix}
    I & A^T \\
    AS^{-2}  & 0
\end{bmatrix} \begin{bmatrix}
    \Delta s \\
    \Delta y
\end{bmatrix} =
\begin{bmatrix}
    0 \\
    \frac{1}{\mu}(b-AS^{-1}e)
\end{bmatrix},
\end{equation}
where $S={\rm diag}(s)$.
Let $\hat{\Delta s}= S^{-2} \Delta s$, we can have the system 
\begin{equation}\label{eq: AS}
\begin{bmatrix}
    S^{2} & A^T \\
    A  & 0
\end{bmatrix} \begin{bmatrix}
    \hat{\Delta s} \\
    \Delta y
\end{bmatrix} =
\begin{bmatrix}
    0 \\
    \frac{1}{\mu}(b-AS^{-1}e)
\end{bmatrix}.
\end{equation}
One can easily verify that $M=\begin{bmatrix}
    S^{2} & A^T \\
    A  & 0
\end{bmatrix}$ is a symmetric positive definite matrix, and so the system \eqref{eq: AS} has a unique solution \cite{roos2005interior}. Given $s$, one can build block-encodings of implementing  matrix $M$ and preparing state $\sigma =\begin{bmatrix}
    0 \\
    \frac{1}{\mu}(b-AS^{-1}e)
\end{bmatrix} $ efficiently, assuming that matrix $A$ stored in QRAM in advance. The general steps of the proposed almost exact QIPM using a short-step framework are described in Algorithm~\ref{alg:AE-QIPM}.

\begin{algorithm}
\caption{Almost Exact QIPM}\label{alg:AE-QIPM}
    \begin{algorithmic}
        \STATE \textbf{INPUT} Dual feasible solution $(y^0,s^0)$, $\mu^0>0$, $0<\theta<1$, and $\delta \left((y^0, s^0), \mu^0\right) < \frac{1}{2}$, where $\delta$ is the proximity measure from \cite{roos2005interior,wu2024quantum}
        \STATE Store $A, b,c $ on QRAM
        \STATE $k \gets 1$
        \WHILE{$\mu > 2^{-2L}$}
        \STATE $(\Delta y^{k}, \hat{\Delta s}^{k} ) \gets $ Solve system \ref{eq: AS} with precision $\epsilon = 2^{-tL}$
        \STATE $y^{k+1}\gets y^{k}+ \Delta y^{k}$
        \STATE $s^{k+1}\gets s^{k}+ (S^{k})^{2} \hat{\Delta s}^{k}$
        \STATE $\mu^{k+1} = (1-\theta) \mu^k$
        \STATE $k \gets k+1$
        \ENDWHILE
    \end{algorithmic}
\end{algorithm}
As we analyze the worst-case complexity, we assume $m=\Ocal(n)$ and matrices are fully dense.
\begin{theorem}\label{theo: IPM}
    Number of iterations for Algorithm~\ref{alg:AE-QIPM} has upper bound $$\Ocal (\sqrt{n}L).$$
\end{theorem}
We prove the theorem in the next section.

\subsection{Proof of Theorem~\ref{theo: IPM}}
Suppose we start with a strictly feasible solution $(x_0,y_0,s_0)$. In the dual logarithmic barrier IPM, we do not compute the value of $x$ and $y$ but they exist. We have
\begin{equation*}
    \begin{aligned}
        Ax_0 = b,\ A^T y_0 + s_0 = c, \ s_0>0.
    \end{aligned}
\end{equation*}
Then we use a quantum subroutine to compute an inexact $\Delta s_0$ with associated error $\xi_1$. After a full Newton step, we have
\begin{equation*}
    Ax_1 = b,\ A^T y_1 + s_1 = c + \xi_1,\ s_1>0.
\end{equation*}
Now we get a feasible solution for problem~$(A,b,c+\xi_1)$. We do another full Newton step, then we have
\begin{equation*}
    \begin{aligned}
        Ax_2 = b,\ A^Ty_2 + s_2 = c+\xi_1 + \xi_2,\ s_2>0.
    \end{aligned}
\end{equation*}
We can keep doing this until we have
\begin{equation*}
    \begin{aligned}
        Ax_k = b,\ A^Ty_k + s_k = c+ \sum_{i=1}^k\xi_i,\ s_k>0.
    \end{aligned}
\end{equation*}
Then, we can rewrite all of them into
\begin{equation*}
    \begin{aligned}
        Ax_j = b,\ A^Ty_j + s_j +\sum_{j+1\leq k}^k\xi_{l} = c + r^k,\ s_j>0,
    \end{aligned}
\end{equation*}
where $r^k = \sum_{i=1}^k\xi_i$.
This implies we obtained a series of feasible iterates for problem $(A,b,c+r^k)$. When their Newton steps are obtained exactly for problem $(A,b,c+r^k)$, then this series converges to an optimal solution for the problem in $\mathcal{O}(\sqrt{n})$ iterations.
However, if their Newton steps are inexact but satisfy the conditions in \cite{wu2024quantum}, the $\sqrt{n}$ complexity still holds.
But these Newton steps are artificial steps because we do not know exactly the errors $\xi_i$. We need to show that the actual Newton steps we inexactly compute are close enough to these artificial Newton steps, and the inexactness is acceptable for the convergence conditions.

In the first iteration, the actual and artificial Newton steps are computed as
\begin{equation*}
    \begin{aligned}
        \Delta s_0 &= -A^T \left( A S_0^{-2} A^T \right)^{-1} \frac{1}{\mu_0} \left( b- \mu_0 A S_0^{-1} e \right) +\xi_1\\
        \Delta \tilde{s}_0 &= -A^T \left( A \tilde{S}_0^{-2} A^T \right)^{-1} \frac{1}{\mu_0} \left( b- \mu_0 A \tilde{S}_0^{-1} e \right),
    \end{aligned}
\end{equation*}
where
\begin{equation*}
    \begin{aligned}
        \tilde{S}_0 = S_0+ r^k.
    \end{aligned}
\end{equation*}
According to \cite{wu2024quantum}, we need 
\begin{equation*}
    \begin{aligned}
        \left\|  \tilde{S}_0^{-1} (\Delta s_0 - \Delta \tilde{s}_0)\right\|_2 \leq 0.1\delta_{\tilde{c}}(\tilde{s}_0, \mu_0),
    \end{aligned}
\end{equation*}
where $\delta_{\tilde{c}}$ is the proximity measure for the perturbed problem $(A, b, \tilde{c})$ with $\tilde{c} = c + r^k$.
This condition can be guaranteed when
\begin{equation}\label{eq: condition iter1}
    \begin{aligned}
        \left\| (S_0 \tilde{S}_0^{-1} (I - S_0 \tilde{S}_0^{-1})) \right\|_2 &\leq 0.033 \delta_{\tilde{c}}(\tilde{s}_0, \mu_0),\\
        \left\| I - (S_0 \tilde{S}_0^{-1})^2 \right\|_2 &\leq 0.033,\\
        \left\| \tilde{S}_0^{-1} \xi_1 \right\|_2 &\leq 0.033 \delta_{\tilde{c}}(\tilde{s}_0, \mu_0).
    \end{aligned}
\end{equation}
Notice that all three conditions can be satisfied by pushing $\xi_i$ to be small as long as $\delta_{\tilde{c}}$ is not zero, which can be inferred by the approximate value of $\delta$. 
We discuss the value of $\xi_i$ later in the section. 
This proves that our inexact Newton step is a feasible inexact Newton step for the perturbed problem. 
Then, according to Theorem 3.3 of \cite{wu2024quantum}, we have the $\mathcal{O}(\sqrt{n}L)$ complexity.

After $\Ocal(\sqrt{n}L)$ iterations, we have an $\tilde{x}>0$ such that 
 \begin{align*}
        A\tilde{x}&=b, \\
        A^T y^k + s^k &= c+r^k,\\
        (\tilde{x})^T s^k &\leq 2^{-2L},
    \end{align*}
where $r^k=\sum_{i=1}^k\xi_i$. It is easy to verify that $(\tilde{x}, y^k, s^k)$ is a $2^{-tL}$-optimal solution for the perturbed problem $(A,b,c+r^k)$, and one can calculate the exact optimal solution by a rounding procedure. It is easy to verify that $\|r^k\|\leq 2^{(1-t)L}$. In the remaining part, we show how we can retrieve an optimal solution of the original problem with a rounding procedure from the optimal solution for the perturbed problem.

It is straightforward to see that $(\tilde{x}, y^k, s^k)$ is in a $2^{(1-t)L}$-neighborhood of the optimal set for the original problem $(A,b,c)$. As the smallest nonzero element of $s^*$ and $x^*$ is greater than $2^{-L}$, using partitions $B$ and $N$ of this solution, by solving a constrained least squares problem, an optimal solution for the original problem can be obtained. For the details of the rounding procedures, refer to Chapter 7 of \cite{wright1997primal}. 

It is worth noting that the rounding procedures are strongly polynomial-time methods. They can also be quantized using quantum linear system solvers; however, we do not explore the cost and implementation details of rounding procedures in this paper, as it is beyond the scope of this paper.

\section{Quantum Subroutine}\label{sec: Sub}
In this section, we analyze the complexity of building and solving system \eqref{eq: AS}. We use the general scheme of the Quantum Tomography framework of \cite{mohammadisiahroudi2024quantum,mohammadisiaroudi2023exponentially}. We assume that we have access to a large enough QRAM, and we store data $A, b, c$ initially on QRAM with worst-case $\Ocal(n^2)$ complexity. At each state, we need to store $s$ on QRAM and build and solve System \eqref{eq: AS} using the iterative quantum linear solver of \cite{mohammadisiaroudi2023exponentially}.
\begin{algorithm}
\caption{Quantum Linear Solver}\label{alg:QLSA}
    \begin{algorithmic}
        \STATE \textbf{INPUT} $(A,b,c)$ stored on QRAM,
        \STATE Store $s$ on QRAM
        \STATE $k \gets 1$
        \STATE $z^k \gets 0$
        \WHILE{$\|z^k - z^{k-1}\|  > 2^{-4L}$}
        \STATE Prepare State $\ket{r^k} = \ket{\sigma - M z^k}$
        \STATE Apply inverse of block encoding of $M$ using QSVT \cite{chakraborty2018power}
        \STATE Extract classical solution $\frac{p^k}{\|p^k\| }= \frac{M^{-1}r^k}{\|M^{-1}r^k\|} $ via Tomography \cite{van2023quantum} with precision $\epsilon = 10^{-2}$
       \STATE Estimate norm of $\|p^k\|$ and $\|r^k\|$
        \STATE $z^{k+1}\gets z^{k}+ \frac{p^{k}}{\|r^k\|}$
        \STATE $k \gets k+1$
        \ENDWHILE
    \end{algorithmic}
\end{algorithm}
At each iteration of Algorithm~\ref{alg:QLSA}, the only classical operation is updating the solution by a vector summation with $\Ocal(n)$ arithmetic operations. In the following, we calculate the cost of quantum operations.
\begin{lemma}\label{lem: block1}
    Given $A$ and $S$ stored on QRAM, the following statements are true:
    \begin{itemize}
        \item We can construct a block-encoding of $M$ using $\Ocal(\text{polylog}(\frac{n}{\epsilon}))$ queries to QRAM.
        \item We can prepare the the state $\ket{r}$ using $\Ocal(\text{polylog}(\frac{n}{\epsilon}))$ queries to QRAM.
        \item We can apply $M^{-1}$ using $\tilde{\Ocal}_{n,\kappa, \frac{1}{\epsilon}}(\kappa\|A\|_F)$ queries to QRAM. \footnote{The $\widetilde{\Ocal}_{\alpha, \beta} \left( g(x) \right)$ notation indicates that quantities polylogarithmic in $\alpha, \beta$ and $g(x)$ are suppressed.}
        \item Norm estimation of $p^k$ and $r^k$ costs $\tilde{\Ocal}_{n,\kappa, \frac{1}{\epsilon}}(\kappa\|A\|_F)$ queries to QRAM.
    \end{itemize}
\end{lemma}
The proof of Lemma \ref{lem: block1} is the direct result of Prepositions 1 to 6 of \cite{augustino2023quantum}. 

\begin{lemma}\label{lem: qta iteration}
    The number of iterations of Algorithm~\ref{alg:QLSA} is at most $\Ocal(L)$.
\end{lemma}
The proof of Lemma \ref{lem: qta iteration} is based on \cite{mohammadisiaroudi2023exponentially}. Additionally, the total complexity of Algorithm~\ref{alg:QLSA} is based on the analysis provided by \cite{mohammadisiaroudi2023exponentially}. 

\begin{theorem}
    Assuming $(A, b, c)$ is stored on QRAM, Algorithm~\ref{alg:QLSA} can find a $2^{-tL}$-precision solution for System~\eqref{eq: AS} with 
    $$\tilde{\Ocal}_{n\kappa L}(n\kappa\|A\|_F)$$
    queries to QRAM.
\end{theorem}

\section{Proposed IR-AE-QIPM}\label{sec: IR}
In this section, we discuss how to use the iterative refinement method (IR) for LO problems to improve complexity as in \cite{wu2024quantum} and provide the full description of our proposed algorithm. The first iterative refinement for LO has been proposed by \cite{gleixner2016iterative}. Mohammadisiahroudi et al. \cite{mohammadisiahroudi2024efficient} first showed that using iterative refinement can improve the complexity of QIPMs w.r.t precision and condition number. Further, in \cite{mohammadisiahroudi2023inexact}, the quadratically convergent iterative refinement scheme was proposed for feasible IPMs. An IR for dual log-barrier QIPM has been developed in \cite{wu2024quantum}.

In \cite{wu2024quantum}, the iterative refinement method for the LO problem works as follows:
\begin{itemize}
    \item[1.] Start with the original problem and solve it to a low accuracy;
    \item[2.] If the accuracy of the original problem is not enough, construct a refining problem using the current iteration values; otherwise, the algorithm halts;
    \item[3.] Solve the refining problem to a low accuracy and update the solution to the original problem; then, go to step 2.
\end{itemize}

In our proposed algorithm, after each solve, we have a feasible solution to a perturbed problem. To use the iterative refinement method, we need to construct a solution to the original problem from the solution to a perturbed problem.
To do so, we need a projection procedure. We use Algorithm~\ref{alg:QLSA} to solve the following problem
\begin{equation*}
    \begin{aligned}
        \min_{y} \|A^T y + s_k - c\|_2,
    \end{aligned}
\end{equation*}
which is equivalent to
\begin{equation*}
    \begin{aligned}
        AA^T y = A(c-s_k).
    \end{aligned}
\end{equation*}
Then we have 
\begin{equation*}
    \begin{aligned}
        s = c- A^T y.
    \end{aligned}
\end{equation*}
According to the argument in the previous section, this $(y, s)$ is feasible for the original problem with a duality gap bounded by twice of the low accuracy.
Then, we can use the IR to refine the solution to high accuracy as in \cite{wu2024quantum}.

To get the full complexity of the proposed algorithm, we discuss the accuracy needed for $\xi_i$. In the first iteration, we need conditions \eqref{eq: condition iter1}. Theoretically, $\delta_{\tilde{c}}$ might be zero, which implies the corresponding Newton system right-hand side is zero. We do not need to solve such Newton systems. Instead, we check the norm of the right-hand side vector. If the norm is too small $(\leq 2^{-4L})$, we update $\mu$ without computing the Newton step. Then, conditions \eqref{eq: condition iter1} can be guaranteed when
\begin{equation*}
    \begin{aligned}
        \|\xi_i\|_2 \leq {\rm poly}\left(\frac{2^{-4L} }{n\kappa_{AS_0^{-1}}} \right) \approx {\rm poly} (2^{-4L}), \ \forall i\in[k].
    \end{aligned}
\end{equation*}
This bound also works for the remaining iterations.
Now, we present the pseudocode of our proposed algorithm and the main theorem.

\begin{algorithm}
\caption{Iteratively Refined Almost Exact QIPM }\label{alg:IR-AE-QIPM}
    \begin{algorithmic}
        \STATE \textbf{INPUT} Dual feasible solution $(y^0,s^0)$, $\mu^0>0$, $0<\theta<1$, $\delta \left((y^0, s^0), \mu^0\right) < \frac{1}{2}$, $\nabla^{(0)} = 1$, $0<\zeta\ll\tilde{\zeta}$
        \STATE Store $A, b,c $ on QRAM
        \STATE $k \gets 1$
        \STATE $(y_1, s_1) \gets$ Solve dual problem with accuracy $\tilde{\zeta}$
        \WHILE{$\nabla^{(k-1)} < \frac{1}{\zeta}$}
        \STATE $\nabla^{(k)} \gets \nabla^{(k-1)}\times\frac{1}{\tilde{\zeta}}$
        \STATE Construct the IR problem as in \cite{wu2024quantum}
        \STATE $(\hat{y}, \hat{s}) \gets $ Solve IR problem with accuracy $\tilde{\zeta}$ and project into proper subspace
        \STATE $y^{k+1}\gets y^{k}+ \frac{1}{\nabla^{(k)}} \hat{y}$
        \STATE $s^{k+1} \gets c- A^T y^{(k)}$
        \STATE $k \gets k+1$
        \ENDWHILE
    \end{algorithmic}
\end{algorithm}

\begin{theorem}[Lemma 13 of \cite{wu2024quantum}]\label{theo: IR iter}
    Algorithm \ref{alg:IR-AE-QIPM} terminates after $\Ocal(\frac{\log (\zeta)}{\log (\hat{\zeta})})$ iterations.
\end{theorem}
For our purpose, we use $\zeta=2^{-tL}$ and $\hat{\zeta}$ is constant. Thus the outer iteration has iteration bound $\Ocal(L)$. We also have a matrix-vector product at each step of this IR scheme with cost $\Ocal(n^2)$ arithmetic operations.

The major challenge in AE-QIPM Algorithm~\ref{alg:AE-QIPM} is that the complexity of the quantum solver depends on the condition number, and the condition number grows in each iteration of AE-QIPM. As in IR Algorithm~\ref{alg:IR-AE-QIPM}, we stop AE-QIPM early at fixed precision. It has been shown that with early termination $\kappa^{(k)}=\Ocal(\kappa_0)$ where $\kappa_0$ is the condition number of the coefficient matrix for $(y^0, s^0)$ and it is constant \cite{wu2024quantum}. Figure~\ref{fig:IR-cond}, adopted from \cite{mohammadisiahroudi2023inexact}, illustrates how the condition number of the linear systems in IR-IF-QIPM behaves for a degenerate LO problem. We remark that IR does not improve the initial condition number.

\begin{figure}
    \centering
    \includegraphics[scale=0.6]{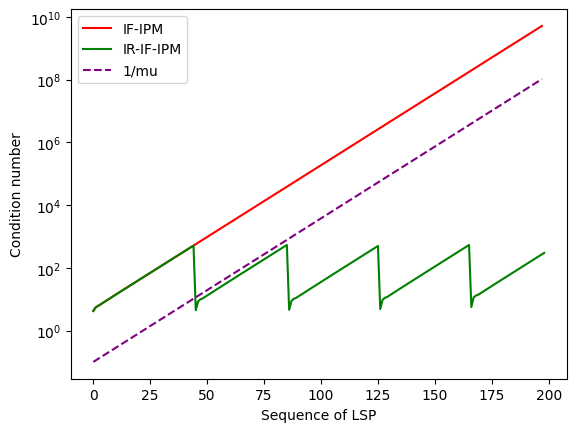}
    \caption{The effect of IR on the condition number of linear systems arising in QIPM for degenerate LO \cite{mohammadisiahroudi2023inexact}.}
    \label{fig:IR-cond}
\end{figure}

\section{Total Complexity}\label{sec: total}
In this section, we put together all the elements discussed in the previous sections to calculate the total worst-case complexity of IR-AE-QIPM.
\begin{theorem}\label{theorem: main}
    Algorithm~\ref{alg:IR-AE-QIPM} produces a $2^{(1-t)L}$ precise optimal solution of the LO problem using at most $$\tilde{\Ocal}_{\kappa_0, n, \|A\|_F}(n^{1.5}L\kappa_0))$$ queries to QRAM and $\mathcal{O}(n^2L)$ classical arithmetic operations.
\end{theorem}

\begin{proof}
    The number of iterations of IR is bounded by $\Ocal(L)$ based on Theorem~\ref{theo: IR iter}. At each iteration, we have $\Ocal(n^2)$ cost of a classical matrix-vector product and the cost of AE-QIPM to solve the refining problem. Additionally, to address $\|A\|_F$ in the complexity, one can initially normalize data by $\|A\|_F$, and consequently, final precision should be increased by $\|A\|_F$, which appears in polylog. The quantum complexity is $\tilde{\Ocal}_{n,L, \|A\|_F}(n^{1.5}\kappa_0L)$ queries to QRAM, and $\Ocal(nL)$ arithmetic operations at each step of AE-QIPM. Thus, the total queries to QRAM is 
    $$\tilde{\Ocal}_{\kappa_0, n}(n^{1.5}L\kappa_0)),$$ 
    and the total number of classical arithmetic operations is bounded by 
    $\mathcal{O}(n^2L).$
\end{proof}

As we can see, the worst-case complexity of IR-AE-QIPM has optimal dimension dependence, $n^2$. For the worst-case scenario, storing the dense matrix $A$ requires $\Ocal(n^2)$ operations. Thus, the end-to-end worst-case complexity of an LO solver with respect to dimension can not be less than $n^2$.

Table~\ref{tab:compelxities} compares the complexity of the proposed IR-AE-QIPM with other classical and quantum IPMs. As we can see, the total complexity of our approaches outperforms previous complexities. In the last line of the table, we show the complexity of the classical counterpart of the IR-AE-IPM using CG to solve the system. As we can see, the total complexity can not be better than $n^{2.5}$ in the classical version. This exhibits a clear quantum advantage compared to other algorithms in the literature. It should be mentioned that the quantum complexity of all QIPMs is the number of queries to QRAM. Without QRAM assumptions, some overheads may appear in complexities, although the quantum central path method of \cite{augustino2023central} is QRAM-free.
\renewcommand{\arraystretch}{1.5}
\begin{table*}[ht]
\caption{Worst-case Complexity of different IPMs for LO}
\label{tab:compelxities}
\centering
\resizebox{\textwidth}{!}{%
\begin{tabular}{cccc}
\hline
Algorithm                & Linear System Solver & Quantum Complexity & Classical Complexity  \\ \hline\hline
IPM with Partial Updates \cite{roos2005interior}            &            Low rank updates          &                    &        $\Ocal(n^{3}L)$                       \\ \hline
Feasible IPM \cite{roos2005interior}            &  Cholesky             &                    &                $\Ocal(n^{3.5}L)$               \\ \hline
II-IPM   \cite{Monteiro2003_Convergence}                          & PCG                   &                    &                $\Ocal(n^{5}L\bar{\chi}^2)$        \\ \hline
Robust IPM \cite{brand2020}                        & Fast Mat-Mul and Partial Update                  &                    &                $\Ocal(n^{w}L)$              \\
\hline\hline
Quantum Central Path \cite{augustino2023central}
 &          Hamiltonian Evolution      &          $\tilde{\Ocal}(n^{3.5}\frac{\omega}{\epsilon})$          &                              \\ \hline

IR-IF-IPM \cite{mohammadisiahroudi2025improvements}                      & PCG                  &                    &    $\tilde{\Ocal}_{\mu^0}(n^{3.5}L\bar{\chi}^2)$     \\ \hline
IR-IF-QIPM \cite{mohammadisiahroudi2023inexact}                      & QLSA+QTA             &             $\tilde{\Ocal}_{n,\kappa_{A}, \|A\|,\|b\|,\mu^0}(n^{1.5}L\kappa_{A}^2\omega^5)$       &          $\tilde{\Ocal}_{\mu^0}(n^{2.5}L)$                  \\ \hline
 IR-IF-QIPM \cite{mohammadisiahroudi2025improvements}                      & Precond+QLSA+QTA             &      $ \widetilde{\Ocal}_{ n, \left\| A \right\|_F, \frac{1}{\epsilon}} ( n^{1.5} L\bar{\chi}^2      ) $              &            $\tilde{\Ocal}_{\mu^0}(n^{2.5}L)$                 \\ \hline
 IPM with approximate Newton steps~\cite{apers2023quantum}     & Q-spectral Approx.     & $\tilde{\mathcal{O}}_{n, \frac{1}{\zeta}}(n^{5.5})$                   & $\tilde{\mathcal{O}}_{\frac{1}{\zeta}}(n^{1.5} )$  \\ \hline
Quantum Dual-log Barrier \cite{wu2024quantum}                                       & QLSA+QTA               & $\widetilde{\mathcal{O}}_{n, \kappa_0, \mu^0,\|A\|_F}\left(n^{1.5} \kappa_0 L \right)$ & $\mathcal{O}( {n}^{2.5}L)$ \\ \hline
 Proposed IR-AE-QIPM                       & IQLSA+Quant Mat-Vec           &             $\tilde{\Ocal}_{n,\kappa_{0},\|A\|_F}(n^{1.5}L\kappa_{0})$       &          ${\Ocal}(n^{2}L)$                    \\ \hline
 Classical IR-AE-IPM                       & CGM             &                 &            ${\Ocal}(n^{2.5}L\kappa_0)$                  \\ \hline
\end{tabular}%
}
\end{table*}

\section{Conclusion}\label{sec: con}
In this paper, we introduce a novel Almost-Exact Interior Point Method framework for solving linear optimization problems. At each iteration, all computations, including matrix-vector products, are performed on a quantum machine, achieving a clear quantum speed-up compared to classical IPMs. To maintain exponentially small errors, we incorporate iterative refinement both within and outside the IPM steps. The overall complexity achieves an optimal scaling of $n^2$.

The main limitation of the proposed method is its reliance on QRAM, whose physical realization remains challenging. However, by using circuit-based QRAM \cite{park2019circuit}, novel QSVTs without block-encoding \cite{chakraborty2025quantum} or adopting a sparse-access input model, one can study the complexity of a QRAM-free version of the proposed algorithm.
Furthermore, a detailed resource estimate, e.g., the analysis of \cite{tu2025towards}, is needed to investigate the practical advantage of the proposed approach.

Another interesting future research direction can be developing a primal-dual Almost-exact QIPM for linear optimization as well as semidefinite optimization problems. Such primal-dual frameworks can be applied to self-dual embedding formulations, which do not require an initial interior solution.

{\hyphenpenalty=100000
\bibliographystyle{IEEEtran}
\bibliography{references}}

@book{wright1997primal,
  title={Primal-Dual Interior-Point Methods},
  author={Wright, Stephen J},
  year={1997},
  publisher={SIAM}
}

@book{bertsimas1997introduction,
  title={Introduction to Linear Optimization},
  author={Bertsimas, Dimitris and Tsitsiklis, John N},
  volume={6},
  year={1997},
  publisher={Athena Scientific Belmont, MA}
}

@book{roos2005interior,
  title={Interior Point Methods for Linear Optimization},
  author={Roos, Cornelis and Terlaky, Tam{\'a}s and Vial, J-Ph},
  year={2005},
  publisher={Springer Science \& Business Media}
}

@article{al2009convergence,
  title={Convergence analysis of the inexact infeasible interior-point method for linear optimization},
  author={Al-Jeiroudi, Ghussoun and Gondzio, Jacek},
  journal={Journal of Optimization Theory and Applications},
  volume={141},
  pages={231--247},
  year={2009},
  publisher={Springer}
}

@article{Monteiro2003_Convergence,
	title={Convergence analysis of a long-step primal-dual infeasible interior-point {LP} algorithm based on iterative linear solvers},
	author={Monteiro, Renato D.C. and O’Neal, Jerome W.},
	journal={Georgia Institute of Technology},
	url={https://bit.ly/3ss0kg1},
	year={2003}
}

@inproceedings{lee2015efficient,
  title={Efficient inverse maintenance and faster algorithms for linear programming},
  author={Lee, Yin Tat and Sidford, Aaron},
  booktitle={2015 IEEE 56th Annual Symposium on Foundations of Computer Science},
  pages={230--249},
  year={2015},
  organization={IEEE}
}

@article{cohen2021solving,
  title={Solving linear programs in the current matrix multiplication time},
  author={Cohen, Michael B and Lee, Yin Tat and Song, Zhao},
  journal={Journal of the ACM (JACM)},
  volume={68},
  number={1},
  pages={1--39},
  year={2021},
  publisher={ACM}
}

@book{nesterov1994interior,
  title={Interior-{Point} {Polynomial} {Algorithms} in {Convex} {Programming}.},
  author={Nesterov, Yurii E and Nemirovskii, Arkadi},
  volume={13},
  year={1995},
  publisher={SIAM},
address={Philadelphia, PA, USA}
}

@article{kerenidis2020quantum,
  title={A quantum interior point method for {LPs} and {SDPs}},
  author={Kerenidis, Iordanis and Prakash, Anupam},
  journal={ACM Transactions on Quantum Computing},
  volume={1},
  number={1},
  pages={1--32},
  year={2020},
  publisher={ACM New York, NY, USA}
}

@article{wu2024quantum,
  title={A quantum dual logarithmic barrier method for linear optimization},
  author={Wu, Zeguan and Sampourmahani, Pouya and Mohammadisiahroudi, Mohammadhossein and Terlaky, Tam{\'a}s},
  journal={arXiv preprint arXiv:2412.15977},
  year={2024}
}

@article{klee1972good,
  title={How good is the simplex algorithm},
  author={Klee, Victor and Minty, George J},
  journal={Inequalities},
  volume={3},
  number={3},
  pages={159--175},
  year={1972},
  publisher={New York}
}

@article{Khachiyan1980,
title = {Polynomial algorithms in linear programming},
journal = {USSR Computational Mathematics and Mathematical Physics},
volume = {20},
number = {1},
pages = {53-72},
year = {1980},
issn = {0041-5553},
doi = {https://doi.org/10.1016/0041-5553(80)90061-0},
url = {https://www.sciencedirect.com/science/article/pii/0041555380900610},
author = {L.G. Khachiyan}
}

@inproceedings{Karmarkar1984_New,
	author={Karmarkar, Narendra},
	title = {A new polynomial-time algorithm for linear programming},
	year = {1984},
	publisher = {Association for Computing Machinery},
	address = {New York, NY, USA},
	booktitle = {Proceedings of the Sixteenth Annual ACM Symposium on Theory of Computing},
	pages = {302–311},
	series = {STOC '84}
}

@article{chambolle2011first,
  title={A first-order primal-dual algorithm for convex problems with applications to imaging},
  author={Chambolle, Antonin and Pock, Thomas},
  journal={Journal of Mathematical Imaging and Vision},
  volume={40},
  pages={120--145},
  year={2011},
  publisher={Springer}
}

@article{applegate2021practical,
  title={Practical large-scale linear programming using primal-dual hybrid gradient},
  author={Applegate, David and D{\'\i}az, Mateo and Hinder, Oliver and Lu, Haihao and Lubin, Miles and O'Donoghue, Brendan and Schudy, Warren},
  journal={Advances in Neural Information Processing Systems},
  volume={34},
  pages={20243--20257},
  year={2021}
}

@article{apers2023quantum,
  title={Quantum speedups for linear programming via interior point methods},
  author={Apers, Simon and Gribling, Sander},
  journal={arXiv preprint arXiv:2311.03215},
  year={2023}
}

@inproceedings{brand2020,
author = {Jan van den Brand},
title = {A deterministic linear program solver in current matrix multiplication time},
booktitle = {Proceedings of the 2020 ACM-SIAM Symposium on Discrete Algorithms (SODA)},
chapter = {},
pages = {259-278},
doi = {10.1137/1.9781611975994.16},
URL = {https://epubs.siam.org/doi/abs/10.1137/1.9781611975994.16},
eprint = {https://epubs.siam.org/doi/pdf/10.1137/1.9781611975994.16},
year = {2020}
}

@techreport{augustino2021inexact,
    author ={Augustino, Brandon and Mohammadisiahroudi, Mohammadhossein and Terlaky, Tam{\'a}s and Zuluaga, Luis F},
    title = {An inexact-feasible quantum interior point method for second-order cone optimization},
    institution = "Lehigh University, 21T-009",
    year = {2021}
}

@inproceedings{van2023quantum,
  title={Quantum tomography using state-preparation unitaries},
  author={van Apeldoorn, Joran and Cornelissen, Arjan and Gily{\'e}n, Andr{\'a}s and Nannicini, Giacomo},
  booktitle={Proceedings of the 2023 Annual ACM-SIAM Symposium on Discrete Algorithms (SODA)},
  pages={1265--1318},
  year={2023},
  organization={SIAM}
}

@article{harrow2009quantum,
  title={Quantum algorithm for linear systems of equations},
  author={Harrow, Aram W and Hassidim, Avinatan and Lloyd, Seth},
  journal={Physical Review Letters},
  volume={103},
  number={15},
  pages={150502},
  year={2009},
  publisher={APS}
}

@article{chakraborty2018power,
  title={The power of block-encoded matrix powers: improved regression techniques via faster {H}amiltonian simulation},
  author={Chakraborty, Shantanav and Gily{\'e}n, Andr{\'a}s and Jeffery, Stacey},
  journal={arXiv preprint arXiv:1804.01973},
  year={2018}
}

@article{childs2017quantum,
  title={Quantum algorithm for systems of linear equations with exponentially improved dependence on precision},
  author={Childs, Andrew M and Kothari, Robin and Somma, Rolando D},
  journal={SIAM Journal on Computing},
  volume={46},
  number={6},
  pages={1920--1950},
  year={2017},
  publisher={SIAM}
}

@article{augustino2023central,
  title={A quantum central path algorithm for linear optimization},
  author={Augustino, Brandon and Leng, Jiaqi and Nannicini, Giacomo and Terlaky, Tam{\'a}s and Wu, Xiaodi},
  journal={arXiv preprint arXiv:2311.03977},
  year={2023}}

@article{gleixner2016iterative,
author = {Gleixner, Ambros M. and Steffy, Daniel E. and Wolter, Kati},
title = {Iterative refinement for linear programming},
journal = {INFORMS Journal on Computing},
volume = {28},
number = {3},
pages = {449-464},
year = {2016},
doi = {10.1287/ijoc.2016.0692}
}

@article{augustino2023quantum,
  title={Quantum interior point methods for semidefinite optimization},
  author={Augustino, Brandon and Nannicini, Giacomo and Terlaky, Tam{\'a}s and Zuluaga, Luis F},
  journal={Quantum},
  volume={7},
  pages={1110},
  year={2023},
  publisher={Verein zur F{\"o}rderung des Open Access Publizierens in den Quantenwissenschaften}
}

@article{mohammadisiahroudi2024efficient,
  title={Efficient Use of Quantum Linear System Algorithms in Inexact Infeasible {IPMs} for Linear Optimization},
  author={Mohammadisiahroudi, Mohammadhossein and Fakhimi, Ramin and Terlaky, Tam{\'a}s},
  journal={Journal of Optimization Theory and Applications},
  pages={1--38},
  year={2024},
  publisher={Springer}
}

@article{mohammadisiahroudi2025improvements,
  title={Improvements to quantum interior point method for linear optimization},
  author={Mohammadisiahroudi, Mohammadhossein and Wu, Zeguan and Augustino, Brandon and Carr, Arielle and Terlaky, Tam{\'a}s},
  journal={ACM Transactions on Quantum Computing},
  volume={6},
  number={1},
  pages={1--24},
  year={2025},
  publisher={ACM New York, NY}
}

@article{wu2023inexact,
  title={An inexact feasible quantum interior point method for linearly constrained quadratic optimization},
  author={Wu, Zeguan and Mohammadisiahroudi, Mohammadhossein and Augustino, Brandon and Yang, Xiu and Terlaky, Tam{\'a}s},
  journal={Entropy},
  volume={25},
  number={2},
  pages={330},
  year={2023},
  publisher={MDPI}
}

@techreport{mohammadisiaroudi2023exponentially,
  title={Exponentially more precise tomography for quantum linear system solutions via iterative refinement},
  author={Mohammadisiaroudi, Mohammad and Augustino, Brandon and Fakhimi, Ramin and Nannicini, Giacomo and Terlaky, Tam{\'a}s},
  year={2023},
  institution={Lehigh University, 23T-007}
}

@article{mohammadisiahroudi2025quantum,
  title={Quantum computing inspired iterative refinement for semidefinite optimization},
  author={Mohammadisiahroudi, Mohammadhossein and Augustino, Brandon and Sampourmahani, Pouya and Terlaky, Tam{\'a}s},
  journal={Mathematical Programming},
  pages={1--40},
  year={2025},
  publisher={Springer}
}

@phdthesis{mohammadisiahroudi2024quantum,
  title={Quantum Computing and Optimization Methods},
  author={Mohammadisiahroudi, Mohammadhossein},
  year={2024},
  school={Lehigh University}
}

@article{mohammadisiahroudi2023inexact,
  title={An inexact feasible interior point method for linear optimization with high adaptability to quantum computers},
  author={Mohammadisiahroudi, Mohammadhossein and Fakhimi, Ramin and Wu, Zeguan and Terlaky, Tam{\'a}s},
  journal={arXiv preprint arXiv:2307.14445},
  year={2023},
note={(Accepted for publication in SIAM Journal on Optimization)}}

@article{tu2025towards,
  title={Towards identifying possible fault-tolerant advantage of quantum linear system algorithms in terms of space, time and energy},
  author={Tu, Yue and Dubynskyi, Mark and Mohammadisiahroudi, Mohammadhossein and Riashchentceva, Ekaterina and Cheng, Jinglei and Ryashchentsev, Dmitry and Terlaky, Tam{\'a}s and Liu, Junyu},
  journal={arXiv preprint arXiv:2502.11239},
  year={2025}
}

@article{chakraborty2025quantum,
  title={Quantum singular value transformation without block encodings: Near-optimal complexity with minimal ancilla},
  author={Chakraborty, Shantanav and Hazra, Soumyabrata and Li, Tongyang and Shao, Changpeng and Wang, Xinzhao and Zhang, Yuxin},
  journal={arXiv preprint arXiv:2504.02385},
  year={2025}
}

@article{park2019circuit,
  title={Circuit-based quantum random access memory for classical data},
  author={Park, Daniel K and Petruccione, Francesco and Rhee, June-Koo Kevin},
  journal={Scientific reports},
  volume={9},
  number={1},
  pages={3949},
  year={2019},
  publisher={Nature Publishing Group UK London}
}

\end{document}